\documentclass[pra, aps, twocolumn, preprintnumbers,superscriptaddress]{revtex4}
\usepackage{hyperref}
\usepackage{verbatim}
\usepackage{amsmath}
\usepackage{latexsym}
\usepackage{revsymb}
\usepackage{yfonts}

\usepackage{color}

\usepackage{natbib}
\usepackage{amsfonts}
\usepackage{amsmath}
\usepackage{amssymb}
\usepackage{amsthm}
\usepackage{graphicx}
\usepackage{bm}
\usepackage{bbm}

\newcommand{\be}{\begin{eqnarray} \begin{aligned}}
\newcommand{\ee}{\end{aligned} \end{eqnarray} }
\newcommand{\benn}{\begin{eqnarray*} \begin{aligned}}
\newcommand{\eenn}{\end{aligned} \end{eqnarray*} }

\newcommand{\bc}{\begin{center}}
\newcommand{\ec}{\end{center}}

\newcommand{\id}{\mathbb{I}}

\newcommand{\tr}{\mathop{\mathrm{tr}}\nolimits}

\newtheorem{theorem}{Theorem}[section]

\newtheorem{lemma}[theorem]{Lemma}

\newtheorem{corollary}[theorem]{Corollary}

\newcommand{\hin}{\mathcal{H}_{\rm in}}

\newcommand{\argmax}{\mathop{\mathrm{argmax}}\nolimits}

\newcommand{\bop}{\mathcal{B}}

\usepackage{amsfonts}

\def\id{\mathbb{I}}

\def\01{\{0,1\}}

\newcommand{\floor}[1]{\lfloor{#1}\rfloor}
\newcommand{\eps}{\varepsilon}
\newcommand{\ket}[1]{|#1\rangle}
\newcommand{\bra}[1]{\langle#1|}
\newcommand{\outp}[2]{|#1\rangle\langle#2|}

\newcommand{\proj}[1]{|{#1}\rangle\langle{#1}|}

\newcommand{\eproj}{\Pi_{\delta}^{\eps}}

\newcommand{\henc}{\mathcal{H}_{\rm enc}}
\newcommand{\anc}{\mathcal{H}_{\rm anc}}
\newcommand{\setX}{\mathcal{X}}
\newcommand{\hmin}{\ensuremath{ {\rm H}}_{\infty}}
\newcommand{\HS}{\ensuremath{ {\rm H}}}

\newcommand{\init}{\ket{ {\rm x_{\max}}}}
\newcommand{\pinit}{\proj{ {\rm x_{\max}}}}
\newcommand{\xmax}{ {x_{\max}} }
\newcommand{\xmin}{x_{\min}}
\newcommand{\pmax}{p_{\xmax}}
\newcommand{\pmin}{p_{\xmin}}

\newcommand{\rhomin}{\rho_{\xmin}}
\newcommand{\rhomax}{\rho_{\xmax}}
\newcommand{\Wmin}{W_{\xmin}}

\newcommand{\Pmin}{P_{\xmin}}

\begin{document}

\title{A time-dependent Tsirelson's bound from limits on the rate of information gain in quantum systems}
% steph: I'm hoping more people would read our paper this way even though the title is somewhat horrible, I admit..

\author{Andrew C. Doherty}
\affiliation{Centre for Engineered Quantum Systems, School of Physics, University of Sydney, Sydney, Australia}
\author{Stephanie Wehner}
\affiliation{Centre for Quantum Technologies, National University of Singapore, 2 Science Drive 3, 117543 Singapore}

\date{\today}

\begin{abstract}
	We consider the problem of distinguishing between a set of arbitrary quantum states in a setting in which the time available to perform the measurement is limited. We provide simple upper bounds on how well we can perform state discrimination in a given time as a function of 
	 either the average energy or the range of energies available during the measurement. We exhibit a specific strategy that nearly attains this bound. 
	Finally, we consider several applications of our result. First, we obtain a time-dependent Tsirelson's bound that limits the extent of the Bell inequality violation that can be in principle be demonstrated in a given time $t$.  Second, we obtain a Margolus-Levitin type bound when considering the special case of distinguishing orthogonal pure states.
\end{abstract}

\maketitle

\section{Introduction}

Entropic measures tell us how much information a quantum register $E$ contains about some classical register $X$ \emph{in principle}.
But just how quickly does this information become available to us? In this little note, we derive bounds on the amount of information available
after a given time $t$. As expected, our bounds depend on the resources we have available in the form of the available energy.

Throughout this paper, we will choose to measure information in terms of the {\it min-entropy}, 
which is the relevant quantity when we consider single-shot experiments and quantum cryptography. As we will
explain in detail below, this measure is directly related~\cite{renato:operational} to the probability of success in state discrimination~\cite{helstrom:detection,holevo,belavkin:optimal,yuen:maxState,surveyDiscr}. As a result, we focus on bounding the probability of success 
in distinguishing states $\{\rho_x\}_{x \in \setX}$ where we are given $\rho_x$ with probability $p_x$. Let $P_{\rm guess}(X|E)_{H,t}$ denote this success probability
after time $t$ when using a particular Hamiltonian $H$ in the measurement process. 
After providing a more careful discussion of the measurement process, we show the following results.

\subsection{Results}

\noindent
\textit{A bound for two states:\ } We first consider the case of only two input states $\rho_0,\rho_1$, for which it is easy to compute the optimal success probability
if we have unlimited time (or resources) available~\cite{helstrom:detection}. We first provide a general bound in terms of the spectrum of the Hamiltonian (Corollary~\ref{cor:traceDist}). 
For the special case of of two equiprobable states ($p_0 = p_1 = 1/2$), this bound simply reads
\begin{align}\label{eq:twoStateSummary}
	P_{\rm guess}(X|E)_{H,t} \leq \frac{1}{2} + \frac{\gamma t \|H\|_\infty D(\rho_0,\rho_1)}{2\hbar}\ ,
\end{align}
where $D(\rho_0,\rho_1)$ is the trace distance between the two states, and $\gamma$ is a small constant. This bound is directly related to our ability
to distinguish two inputs states given an unlimited amount of time, where the best measurements gives us~\cite{helstrom:detection}
\begin{align}
	P_{\rm guess}(X|E) = \frac{1}{2} + \frac{D(\rho_0,\rho_1)}{2}\ .
\end{align}
We proceed to show that our bound is nearly tight up to a constant factor (Theorem~\ref{thm:attaining}) by providing an explicit measurement strategy. 
Finally, we prove a bound in terms of the average energies of the input states (Theorem~\ref{thm:avgEnergyTwoStates}). However, this bound does not compare
as easily to the case of unlimited time. 

\smallskip
\noindent
\textit{A bound for many input states:\ } When considering the case of an arbtirary number of input states $\rho_0,\ldots,\rho_{N-1}$ it is difficult to compute
the maximum success probability even in the case of unlimited time. In particular, no general closed form expression is known -- only for the case of single qubit encodings does
there exist a way to construct the optimal measurements geometrically~\cite{barbara:qubit}. 
In general, we can only approximate the optimal 
measurements numerically~\cite{yuen:maxState, eldar:sdpDetector,eldar:sdp,werner:iterate,jezek:iterate,jezek:iterate2,tyson:oldIterate,tyson:newIterate}, 
or resort to bounds on the success probability~\cite{belavkin:optimal,belavkin:radio,mochon:pgm,wootters:pgm,tyson:pgm,tyson:estimates,deepthi:pi}.
As such, it becomes harder to relate $P_{\rm guess}(X|E)_{H,t}$ to case of unlimited time. We hence provide a general bound in terms of the average energies alone. 
In particular, we show (Theorem~\ref{thm:manyStates}) that
\begin{align}
	P_{\rm guess}(X|E) \leq p_{\xmax} + \sum_{x=0}^{N-1} p_x \tr\left(H\rho_x\right)\ ,
\end{align}
where $\xmax$ is the smallest $x \in \{0,\ldots,N-1\}$ such that $p_{\xmax} \geq p_x$ for all $x$.

\smallskip
\noindent
\textit{Applications:\ } Finally, we discuss two applications of our bound. 
The first is to the study of Bell inequalities~\cite{bell}. Typically, we care about determining the maximum quantum violation of such inequalities. 
In contrast, we ask what is the maximum violation that can be achieved in a fixed amount of time. When considering such inequalities as games between two players Alice and Bob (see Section~\ref{sec:game}), the ''amount'' of quantum violation is determined by the probability $p_{\rm win}$ that the players win the game maximized over all states and measurements. For the CHSH inequality~\cite{chsh}, we have that classically
\begin{align}
	p_{\rm win} \leq \frac{3}{4}\ 
\end{align}
for any strategy of Alice and Bob.
However in quantum mechanics there exists a strategy that achieves
\begin{align}
	p_{\rm win} = \frac{1}{2} + \frac{1}{2\sqrt{2}}\ ,
\end{align}
which is optimal~\cite{tsirel:original}. Here, we show (Corollary~\ref{cor:tsirel}) that if we demand answers from Alice and Bob after time $t$
\begin{align}
	p_{\rm win} \leq \frac{3}{4} + \frac{\gamma t \|H\|_\infty}{\hbar\sqrt{2}}\ ,
\end{align}
where $H$ is Bob's Hamiltonian involved in the measurement process, and $\gamma$ is a small constant. 
We will also see that to achieve Tsirelson's bound, Alice and Bob need time at least
\begin{align}
	t \geq \frac{\hbar}{\gamma \|H\|_\infty}\ .
\end{align}
Our bounds tell us that there does indeed a fundamental time that is needed to establish non-local correlations of a 
certain strength. We will discuss these bounds in detail in Section~\ref{sec:tsirel}.

As a second application, we use our bound to obtain a form of the Margolus-Levitin theorem~\cite{mlTheorem} which
provides us with a lower bound on how much time it takes to transform a pure state into an orthogonal state. Since the Margolus-Levitin theorem provides a bound on the
speed of evolution, it clearly provides a bound on the minimum amount of time that is required to obtain the optimal (time unlimited) success probability for state discrimination.
Yet, note that we are interested in bounding $P_{\rm guess}(X|E)_{H,t}$ even for shorter periods of time. We will discuss the relation of our work and the Margolus-Levitin theorem
in detail in Section~\ref{sec:ML}.

\subsection{Related work}

Next to the Margolus-Levitin theorem~\cite{mlTheorem}, our work is related to 
several bounds~\cite{levitin1,pendry} on how fast information can be transmitted in principle
given energy constraints (see~\cite{bekenstein:survey} for a survey of results). These bounds generally consider
the von Neumann entropy as a measure of information and are concerned with determining the capacity for sending 
information as a function of energy. That is, they consider how fast we could convey information in the best possible way. In contrast, 
we consider the case of arbitrary encodings $\rho_x$, which may not be optimal to transmit classical information.
In fact, even in the case of unlimited time the probability that we can reconstruct $x$ from $\rho_x$ could be 
very small. Our setting also differs in the sense that we focus solely on extracting classical information into a
classical register in a sense that we will make precise below.

Our work is also related to several previous papers~\cite{lazyStates,gogolin:mthesis,hayashi:purity} that study the rate of change in 
entropies of a system that is in contact with an environment. Again, our work is a somewhat different flavor since we are
interested in extracting classical information, and our bounds furthermore involve average energies, rather than 
the largest energy $\|H\|_\infty$ of the (interaction) Hamiltonian $H$ alone. 

\section{Gaining classical information}

\subsection{Quantifying information}

Let us now consider more formally what we mean by gaining classical information encoded in a quantum system. 
Imagine that there is some finite set $\setX$ of possible classical symbols to be encoded.
For any symbol $x \in \setX$, we thereby use $\rho_x \in \bop(\henc)$ to denote its encoding into a quantum state
on the system $\henc$. We also refer to $\henc$ as the \emph{encoding space}. Our a priori ignorance about the classical
information $x$ is captured by the probability distribution $p_x$ according to which the encoding space is prepared
in the state $\rho_x$. 

Throughout, we quantify how much information we have about $x$ given access to the encoding space $\henc$ in terms of
the min-entropy~\cite{renato:operational}
\begin{align}
	\hmin(X|E) := - \log P_{\rm guess}(X|E)\ ,
\end{align}
where 
\begin{align}
	P_{\rm guess}(X|E) := \sup_{\substack{\forall x M_x \geq 0\\ \sum_{x \in \setX} M_x = \id}} \sum_{x \in \setX} p_x \tr\left(M_x \rho_x\right)\ ,
\end{align}
is the probability that we guess $x$, maximized over all possible measurements on the encoding space.
Finding the optimal measurement is known as state discrimination and can be done using semidefinite programming~\cite{yuen:maxState,eldar:sdpDetector}.
The min-entropy accurately measures information in a cryptographic setting~\cite{renato:diss}, and for single shot experiments. This is in contrast to the 
von Neumann entropy which is concerned with the asymptotic case of a large number of identical experiments. 

The min-entropy and the von Neumann entropy 
can be arbitrarily different, as is easily seen by considering the example where the encoding is trivial, that is, $\rho_x = \rho_{x'}$ for all $x$ and $x'$.
The strategy that maximizes the guessing probability $P_{\rm guess}(X|E)$ is then simply given by outputting the most likely symbol, i.e., $\hmin(X|E) = 
- \log \max_x p_x$~\footnote{In analogy to the von Neumann entropy, $\hmin(X|E) = \hmin(X)$ since $\hmin(X) := - \log \max_x p_x$.}, 
and the conditional von Neumann entropy obeys $\HS(X|E) = \HS(X) = - \sum_x p_x \log p_x$. Consider now $\Sigma = \{0,1\}^n$ to be the set of bitstrings of length $n$
and suppose the all '0' string occurs with probability $p_{0^n} = 1/2$, and with probability $1/2$ any of the remaining strings occurs with equal probability.
Clearly, we have $\hmin(X|E) = 1$, whereas $\HS(X) \approx n/2$. That is, the von Neumann entropy can be very large, even if there is one symbol that
occurs with extremely high probability. We will remark on the rate of information extraction from a quantum system in terms of the von Neumann entropy later on, but focus on the single shot case given by the min-entropy, 
or equivalently the probability of error in state discrimination.

\subsection{Producing a classical output}

To determine how quickly we can acquire classical information, we first need to specify what it means to output classical information from a measurement.
Here, we model this process with the help of an additional `classical' ancilla system $\anc$ that contains the output. A classical system is associated with a fixed basis, which without
loss of generality we take to be the computational basis. Preparation and measurement of a classical system can only be done in this basis, which intuitively corresponds
to the idea of storing classical information: 
The ancilla can be prepared in any state of the fixed basis, and is subsequently measured in this basis after time $t$.
% steph: I explicitely removed the word ``can be measured'' as this would allow to store quantum information, but cannot, lah
The information contained in this register captures
the notion of a classical probability distribution over the basis elements. 

We model the process of state discrimination as follows. The problem is to discriminate between $N$ states $\rho_x$ on the encoding space $\henc$, where $N$ is the number of possible classical symbols. At the beginning of the experiment the ancilla system is initialized to 
the symbol occurring with the largest probabily $\init$ where
\begin{align}
x_{\max} := \argmax_x p_x\ .
\end{align}
This initial condition captures the distinguisher's apriori knowledge: recall without access to the quantum register $\hmin(X|E) = - \log p_{\xmax}$. 
If the there are multiple classical symbols with the same value $\pmax$, we take the smallest one in lexicographic order.
We will discuss the choice of initial state in detail below.
The ancilla system has total dimension $d_{\anc} = N$ and the other directions correspond to the classical symbols $x$. The experimenter implements a unitary $U$ on $\henc \otimes \anc$ during a specified time $t$. At this point the ancilla system is passed to a referee who will decide whether information has been gathered successfully by measuring $\anc$ in the computational basis, using measurement operators
\begin{align}
	P_x := \proj{x}\ ,
\end{align}
where the subscript $x$ denotes the corresponding classical output. 
Hence the success probability of correctly identifying the state $\rho_x$ using this procedure when the ancilla was initially in the state
$\init$ is given by
\begin{align}\label{eq:succProb}
	\tr\left( (\id_{\henc} \otimes P_x) U (\rho_x \otimes \pinit|)U^\dagger\right)\ .
\end{align}
See Figure~\ref{fig:ancillaUse} for a schematic depiction of this process.
Note that  the ancilla is measured by the referee at no time cost. This is a natural assumption in our setting where we imagine that the final information is extracted by a referee who is not limited by any energy constraints.  
Such a referee naturally arises in, for example, the setting of Bell inequalities which we consider later. 
We will from now on assume that measurements producing classical outcomes
are always performed this way.

\begin{figure}[h]
	\begin{center}
		\includegraphics[scale=1]{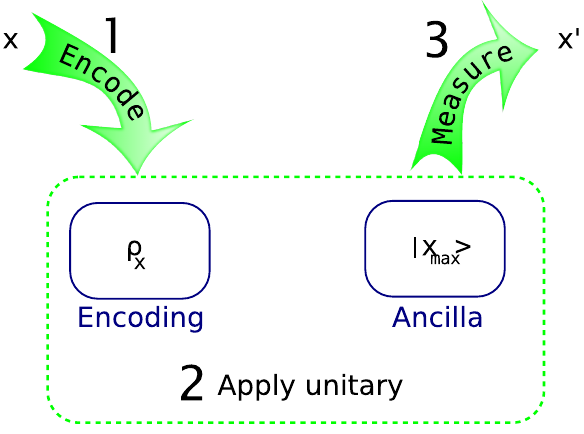}
	\caption{Our protocol for distinguishing quantum states in finite time. First, the encoding register is placed into an encoding $\rho_x$ of the classical symbol $x$ chosen with probability $p_x$. 
	The ancilla is intialized in the state $\init$.
	Second, we can perform
	a unitary interaction $U = \exp(-i Ht/\hbar)$ for time $t$ between the encoding and the ancilla register. Finally, the ancilla register is measured by the 
	referee in the computational basis to determine a guess $x'$ for $x$. If $x'=x$, then we successfully recovered the classical information. In the setting of Bell inequalities considered later on, the ancilla register is simply the message returned to the referee.}
	\label{fig:ancillaUse}
	\end{center}
\end{figure}

To bound how much min-entropy we have after time $t$, our goal is to place bounds on the success probability 
in terms of the unitary
\begin{align}
	U = \exp\left(-\frac{iHt}{\hbar}\right)\ ,
\end{align}
that is, in terms of the interaction Hamiltonian 
\begin{align}
	H = \sum_n E_n \proj{E_n}
\end{align}
and the time $t$. Throughout, we will assume that $H \geq 0$ and that the lowest energy level is in fact $E_0 = 0$. 
Any other Hamiltonian differs from such an $H$ by a term proportional to the identity, which does not contribute 
to the speed of information gain.
%, or, when considering the inverse operation $U^\dagger$, to the speed of information loss.
We explictly chose not to use the common convention $\hbar=1$ to make it easier to draw 
comparisons to the Margolus-Levitin theorem~\cite{mlTheorem} later on.

Before turning to our actual bounds, let us first introduce some additional notation which we will refer to throughout the paper. 
We will use 
\begin{align}
	\tilde{\rho}_x := \rho_x \otimes \pinit\ ,
\end{align}
to denote the combined state consisting of the input state $\rho_x$ on the encoding space, and the initial state of the 
ancilla $\pinit$.
We also write
\begin{align}
	R := U - I = \sum_n (\exp(-i E_n t/\hbar) - 1) \proj{E_n}\ .
\end{align}
Furthermore, it will be convenient to rewrite the success probability~\eqref{eq:succProb} in terms of measurement operators
\begin{align}
	\label{eq:measurementOperators}
	M_x := U^\dagger(\id \otimes P_x)U = \id \otimes P_x + W_x\ ,
\end{align}
as $\tr(M_x \tilde{\rho}_x)$, where
\begin{align}\label{eq:Wdef}
	W_x = W_{x}^{1} + W_{x}^{2}\ ,
\end{align}
and
\begin{align}
	W_{x}^{1} &:= (\id \otimes P_x)R + R^\dagger (\id \otimes P_x)\ ,\\
	W_{x}^{2} &:= R^\dagger (\id \otimes P_x) R\ .
\end{align}
The average success probability for a particular Hamiltonian $H$ and time $t$ can now be written as
\begin{align}
	P_{\rm guess}(X|E)_{H,t} := \sum_{x \in \setX} p_x \tr\left(M_x\tilde{\rho}_x\right)\ .
\end{align}

\section{Time vs. information gain}

We are now ready to derive our bounds. For simplicity, we will outline how this can be done for the case of two equiprobable states, and merely state
our general result. Precise statements as well as a detailed derivation can be found in the appendix. 

\subsection{An upper bound to $P_{\rm guess}(X|E)$}

We now first derive an upper bound to the guessing probability. For the case of two equiprobable states (i.e., $N= 2$ and $p_x = 1/2$ for all $x \in \setX$,
such bounds are easy to obtain when we allow unlimited time (or energy). In particular, it is well known that in this case the success probability is given
by~\cite{helstrom:detection}
\begin{align}\label{eq:distStandard}
	P_{\rm guess}(X|E) := \frac{1}{2} + \frac{D(\rho_0,\rho_1)}{2}\ ,
\end{align}
where $D(\rho_0,\rho_1) = \frac{1}{2} \|\rho_0 - \rho_1\|_1$ is the trace distance of the two states. 
Let us now consider what happens in our time limited scenario for a particular interaction Hamiltonian $H$.
First of all, recall that for two equiprobable states, the ancilla is initialized to the smallest value
\begin{align}
\init = \ket{0}\ .
\end{align}
For two states, 
the success probability $P_{\rm succ}$ averaged over the choice of input state, using the measurement given by operators $M_1$ and $M_0 = \id - M_1$
from~\eqref{eq:measurementOperators}, can now be expressed as
\begin{align}
	&P_{\rm succ}(X|E)_{H,t} \\
	&= \frac{1}{2}\left[ \tr\left(M_0 \tilde{\rho}_0\right) + \tr\left(M_1 \tilde{\rho}_1\right)\right]\nonumber\\
	&=\frac{1}{2}\left[1 + \tr\left(M_1(\tilde{\rho}_1 - \tilde{\rho}_0)\right)\right]\\
	&=\frac{1}{2}\left[1 + \tr\left(\rho_1 - \rho_0\right)\tr\left(P_1 \pinit\right) + \nonumber \right.\\
	&\qquad \left.\tr\left(W_1(\tilde{\rho}_1 - \tilde{\rho}_0)\right)\right]\\
	&=\frac{1}{2} + \frac{\tr\left(W_1(\tilde{\rho}_1 - \tilde{\rho}_0)\right)}{2}\ ,\label{eq:finalPSUCC}
\end{align}
where 
the fourth equality follows immediately from the fact that $P_1 \pinit = \pinit P_1 = 0$.
Let us now upper bound the term involving $W_1$. Again using that $P_1 \pinit =  \pinit P_1 = 0$, we have
\begin{align}
	\tr\left(W_{1}^{1}(\tilde{\rho}_1 - \tilde{\rho}_0)\right) = 0\ .
\end{align}
Define $\tilde{A} := \tilde{\rho}_1 - \tilde{\rho}_0$, and consider its diagonalization $\tilde{A} = \sum_j \lambda_j \proj{u_j}$.  
Let $\tilde{A}^+ := \sum_{j,\lambda_j \geq 0} \lambda_j \proj{u_j}$ and $\tilde{A}^- := \tilde{A} - \tilde{A}^+$.
Using the fact that $R \cdot R^\dagger$ is a positive map~\cite{bathia:posBook} and $0 \leq \id \otimes P_x \leq \id$, we can now bound the
term 
\begin{align}
	\tr\left(W_{1}^{2} \tilde{A}\right) &\leq \tr\left(R \tilde{A}^+ R^\dagger\right)\\
	&\leq 2 \sum_n (1 - \cos(t E_n/\hbar)) \bra{E_n} \tilde{A}^+ \ket{E_n}. 
\end{align}

Substituting back into our original bound (\ref{eq:finalPSUCC}) gives us
\begin{equation}
P_{\rm succ}(X|E)_{H,t}\leq \frac{1}{2} +\sum_n (1 - \cos(t E_n/\hbar)) \bra{E_n} \tilde{A}^+ \ket{E_n}\ . \label{eq:protoBound}
\end{equation}
This is the basic inequality that we can use, along with some restriction on the allowed energies $E_n$, to bound the success probability for state discrimination in time $t$. In the rest of the paper we will apply this in two main settings, bounded maximum energy, and bounded average energy.

\subsubsection{A bound in terms of the maximum energy}
From~\eqref{eq:protoBound}, we can immediately obtain a bound on the success probability for state discrimination in terms of the maximum energy $\|H\|_\infty$ of the coupling Hamiltonian $H$. ($\|H\|_\infty$ is just the largest eigenvalue of $H$.) This bound is attractive since it is simple to derive and has the appealing feature that it involves the 
trace distance between the two states, and is thus directly related to the probability that we distinguish the two states given an unlimited 
amount of time.
However, there are many systems of physical interest where the maximum energy of system states is effectively unbounded. Even though we may without loss of generality assume that the spectrum is bounded for a particular set of input states (see appendix), 
this bound is nevertheless quite unsatisfying in these situations since it can be very weak. For this reason, we use the fundamental inequality~\eqref{eq:protoBound} in the next section to derive a bound on the success probability that depends only on the average energy.

Note that since $\tr(\tilde{A}^+) = D(\tilde{\rho}_0,\tilde{\rho}_1) = D(\rho_0,\rho_1)$
we immediately obtain that the success probability obeys
\begin{align}
	&P_{\rm succ}(X|E)_{H,t} \\
	&\leq \frac{1}{2} + (1 - \cos(t C_{\max}/\hbar)) D(\rho_0,\rho_1)\ ,\nonumber
\end{align}
where $C_{\max} = \argmax_{E_n} (1 - \cos(t E_n/\hbar))$. If $t E_{n}/\hbar \leq 1$ for all $n$, then this upper bound simply reads
\begin{align}\label{eq:spectrumBound}
	&P_{\rm succ}(X|E)_{H,t} \\
	&\leq \frac{1}{2} + (1 - \cos(t \|H\|_\infty/\hbar)) D(\rho_0,\rho_1)\ ,\nonumber
\end{align}
which will be useful for comparison below.
For larger values of $t E_n/\hbar$ it is easy to see that
\begin{align}\label{eq:simpleTrace}
	P_{\rm succ}(X|E)_{H,t} \leq \frac{1}{2} + \frac{\gamma t \|H\|_\infty D(\rho_0,\rho_1)}{2 \hbar}\  ,
\end{align}
where
\begin{align}
	\gamma := \left\{ \begin{array}{ll}5/\pi & \mbox{if } 1 < t E_n/\hbar < 4\ ,\\
	3/\pi & \mbox{otherwise}\ .
	\end{array}\right.
\end{align}

\subsubsection{A bound in terms of the average energy}

A sometimes more satisfying bound can be obtained in terms of the average energy. Note that we can upper bound~\eqref{eq:protoBound}
as
\begin{align}
	\frac{1}{2} + \frac{1}{2}\frac{\gamma t}{\hbar} \sum_n E_n \bra{E_n} \tilde{A}^+ \ket{E_n}\ ,
\end{align}
and hence we may use the fact that 
\begin{align}
	\tilde{A}^+ &= \frac{1}{2}(\tilde{A}^+ - \tilde{A}^-) + \frac{1}{2}(\tilde{A}^+ + \tilde{A}^-)\\
	&\frac{1}{2}(\tilde{\rho}_1 - \tilde{\rho}_0) + \frac{1}{2}|\tilde{\rho}_1 - \tilde{\rho}_0|\ ,
\end{align}
to obtain
\begin{align}
	&P_{\rm succ}(X|E)_{H,t}\\
	&\leq \frac{1}{2} + \frac{ \gamma t \left(\tr(H|\tilde{\rho}_1 - \tilde{\rho}_0|) + \tr(H\tilde{\rho}_1) - \tr(H\tilde{\rho}_0)\right)}
	{4 \hbar} \ .\nonumber
\end{align}
Now, the asymmetry between the labels $0$ and $1$ is inessential. The bound is true if we swap the two state labels, as may be seen by repeating the above derivation swapping the role of the two state labels. Averaging these two bounds we find the following symmetric bound
\begin{align}\label{eq:avg2States}
	P_{\rm succ}(X|E)_{H,t}
	\leq \frac{1}{2} + \frac{ \gamma t \tr(H|\tilde{\rho}_1 - \tilde{\rho}_0|)}
	{4 \hbar} \ .
\end{align}
This bound should be compared with the bound~\eqref{eq:spectrumBound} in terms of the maximum energy in which the trace distance appears. The quantity on the right hand side of~\eqref{eq:avg2States} is loosely an energy-weighted trace distance. Whereas this bound is certainly stronger for a particular choice of $H$, it does not any longer bear an obvious quantitative relation to the Helstrom bound in terms of the trace distance. 
In deriving~\eqref{eq:avg2States} we have made use of the knowledge of the optimal measurements for distinguishing a pair of states. This is no longer possible in more complicated cases, even where unlimited time is allowed~\cite{surveyDiscr}. We can weaken the bound somewhat, using the fact that $\rho,H\geq 0$ to obtain a bound explicitly in terms of the average energy as follows
\begin{align}\label{eq:avg2States2}
	&P_{\rm succ}(X|E)_{H,t}\\
	&\leq \frac{1}{2} + \frac{ \gamma t \tr[H(\tilde{\rho}_0 + \tilde{\rho}_0)]}
	{4 \hbar} \ .\nonumber
\end{align}
So we see that the average energy of the joint system and ancilla place a bound on the success probability of state discrimination, as claimed.
This bound may be generalized easily to the case of more than two classical symbols and an aribtrary distribution $\{p_x\}_x$. 
We show in the appendix that
\begin{theorem}
	Suppose $H \geq 0$. Then the probability of distinguishing $\rho_0,\ldots,\rho_{N-1}$ given with probabilities
	$p_0,\ldots,p_{N-1}$ using the Hamiltonian $H$ obeys
	\begin{align}
		&P_{\rm succ}(X|E)_{H,t}\\
		&\leq \pmax +\frac{\hat{\gamma}t}{\hbar} \sum_{x=0}^{N-1} p_x \tr\left(H\tilde{\rho}_x\right)\ ,\nonumber
	\end{align}
	where
	\begin{align}
		\hat{\gamma} := \left\{ \begin{array}{ll}5/\pi & \mbox{if } \forall E_n, 1 < t E_n/\hbar < 4\ ,\\
			3/\pi & \mbox{otherwise}\ .
		\end{array}\right.
	\end{align}
\end{theorem}

Note that the term $\sum_{x} p_x \tr\left(H\tilde{\rho}_x\right)$ is the energy of the encoding and ancilla register
averaged over the choice of input symbols.

\subsection{A lower bound on $P_{\rm guess}(X|E)$}

We now exhibit a specific measurement strategy for two equiprobable states, which attains our upper bound up to a constant factor. 
We again focus on the case of two possible input states, as for the general setting there is no analytic procedure
of obtaining the optimal measurements even in the setting of unlimited time. Our construction for two states will make explicit 
use of this optimal measurement. 

Let $A = \rho_1 - \rho_0$. It is well known~\cite{helstrom:detection} that the optimal distinguishing measurement
in the time unlimited case without the use of an ancilla
is given by $\{\Pi_{A^+},\Pi_{A^-}\}$, where $\Pi_{A^+}$ and $\Pi_{A^-}$ are projectors on the positive and negative
eigenspace of $A$ respectively. To construct our Hamiltonian $H$, let us diagonalize 
$A = \sum_j \lambda_j \proj{u_j}$, and define $A^+ := \sum_{j,\lambda_j \geq 0} \lambda_j \proj{u_j}$ 
	and $A^- :=  A^+-A$.
Consider the operator
\begin{align}
	\hat{H} &:= \Pi_{A^-} \otimes \id
	+\Pi_{A^+} \otimes (\outp{0}{1} + \outp{1}{0})\ .
\end{align}
Clearly, $\hat{H}$ is Hermitian and unitary, and hence has eigenvalues $\pm 1$.
In fact, $\hat{H}$ is the unitary we would use to achieve the optimum distinguishing probability if we were unconcerned with time.
We now define a Hamiltonian $H$ 
\begin{align}\label{eq:achieveH}
	H := E_{\rm max} (\hat{H} + \id)/2\ .
\end{align}
For comparison with our upper bound of~\eqref{eq:spectrumBound} $H$ obeys the condition $H \geq 0$ and has largest eigenvalue equal to $E_{\rm max} = \|H\|_\infty$.
A simple calculation provided in the appendix shows that for our choice of $H$ we have
\begin{align}
	P_{\rm succ}(X|E)_{H,t} = \frac{1}{2} + \frac{1}{4}(1 - \cos(t \|H\|_\infty/\hbar)) D(\rho_0,\rho_1)\ 
\end{align}
which gives a lower bound to $P_{\rm succ}(X|E)$ maximized over
all possible $H$ in time $t$. This bound matches the upper bound of~\eqref{eq:spectrumBound} up to a factor of $1/4$.

Note that $\hat{H}$ effectively implements a variant of the controlled-NOT (c-NOT) operation on the encoding space and the ancilla. 
For more than two inputs states, one could construct a similar $\hat{H}$ implementing a controlled addition mod $N$ on the ancilla, as long as the optimum distinguishing measurement in the case of unlimited time is a projective measurement on the encoding space.
This would give a similar relation between time and the original probability of distinguishing the given states. 
However, it is known that there do exist choices of encodings $\rho_x$ such that the optimum measurement is not projective, and hence we omit
this restricted form of generalization.

\section{Applications}

Let us now consider several applications of our simple bound. 

\subsection{Minimum distinguishing time and the Margolus-Levitin theorem}\label{sec:ML}

The first application we are interested in, is a return to our initial question: Just how quickly can we acquire information? 
That is, what is the minimum time needed to extract classical information encoded in a quantum system?
Note that with the Hamiltonian $H$ in the lower bound for two equiprobale states, there does indeed exist a way to optimally distinguish 
the two states in time $t = \hbar \pi/\|H\|_\infty$. However, since there is a small gap to our upper bound it would be an open question,
% steph yes would be, because it's not, its our upper bound that is crap which is shown by margolus levitin below
whether it is possible to achieve the same in an even shorter amount of time.

\subsubsection{Minimum time}
Yet, note that our upper bounds on $P_{\rm succ}(X|E)_{H,t}$ can also be understood as lower bounds on the time required to optimally distinguish the given
states, retrieving the maximum amount of information from the encoding. Let us first consider our most general bound for large $\setX$.
We have that if we can distinguish optimally in time $t_{\rm distinguish}$ our upper bound must
be at least as large as the optimum $P_{\rm guess}(X|E)$. That is, 
\begin{align}
	\pmax + \frac{\gamma t_{\rm distinguish}}{\hbar}\sum_{x=0}^{N-1} p_x \tr\left(H\tilde{\rho}_x\right) \geq P_{\rm guess}(X|E)\ ,
\end{align}
and hence
\begin{align}\label{eq:minTime}
	t_{\rm distinguish} \geq \frac{(P_{\rm guess}(X|E) - \pmax) \hbar}{\gamma \sum_{x=0}^{N-1} p_x \tr\left(H\tilde{\rho}_x\right)}\ .
\end{align}

\subsubsection{Margolus-Levitin theorem}

Let us now consider the special case where two equiprobable input encodings are perfectly distinguishable.
That is, $\rho_0 = \proj{0}$ and $\rho_1 = \proj{1}$.
Our task is now quite simple: We merely wish to turn the state $\ket{1}\ket{0}$ of the encoding and ancilla system to the state
$\ket{1}\ket{1}$, that is, we wish to transform one vector into its orthogonal. Note that given unlimited time (or energy) we can
succeed perfectly at this task and hence $P_{\rm guess}(X|E) = 1$. From~\eqref{eq:minTime} we thus have
\begin{align}
	t_{\rm distinguish} &\geq \frac{\hbar}{2\gamma \tr\left(H\tilde{\rho}_1\right)}\ .
\end{align}
Our bound can hence also be understood as putting a limit on the time that it takes to turn a state vector to its orthogonal (on the ancilla),
given some additional resource (the encoding register).

A bound on the minimum time that it takes to turn a vector into its orthogonal is indeed known as the Margolus-Levitin theorem~\cite{mlTheorem}. 
In particular, their bound applied to our situation involving both the encoding and the ancilla register gives
\begin{align}\label{eq:ML}
	t_{\rm ML} & \geq \frac{\hbar \pi}{2 \tr\left(H\tilde{\rho}_1\right)}\ .
\end{align}
Such a bound had previously only been derived from the time-energy uncertainty principle where instead of the average energy, we have the energy spread, i.e, 
the difference in the largest and smallest eigenvalue of the Hamiltonian (see~\cite{lloyd:margolus} for a review of history).
The Margolus-Levitin theorem has been used to place bounds on the 
fundamental speed of computation~\cite{lloyd:margolus}, and was even slightly improved for some special cases~\cite{mlImprovement}. 
Note however that for the Hamiltonian constructed in~\eqref{eq:achieveH} we have $\tr\left(H\tilde{\rho}_1\right) = E_{\rm max}/2$ and
hence the bound provided by Margolus-Levitin is in fact tight as we know that~\eqref{eq:achieveH} lets us achieve the optimum success probability in time
$t = \hbar \pi/E_{\rm max}$. This shows that it is our upper, rather than our lower bound that can be improved.

Since we have $\gamma = 3/\pi$ or $\gamma = 5/\pi$ depending on the parameters, our bound is slightly worse than the Margolus-Levitin bound
which stems from our somewhat crude bound on $(1 - \cos(t E_n/\hbar))$. 
Note, however, that our bound
considers a more specialized situation, namely turning the ancilla to its orthogonal given the encoding, but in turn applies to any kind of input states.

That we obtain a Margolus-Levitin type theorem as a side effect of our analysis is not very surprising: Clearly, the speed of dynamical evolution
places a bound on how quickly we can transfer information from one system into the other. In turn however, note that a bound on how quickly transformation
can be transferred does translate into bounds on the speed of evolution as well and one can think of the speed of dynamical 
evolution when applied to a computation~\cite{lloyd:margolus} as being limited by how quickly one can transfer the necessary information required for the subsequent
stage of computation. 

\subsection{Time-dependent Tsirelson-bound}\label{sec:tsirel}

As another example on how our bound can be used we will derive a time-dependent Tsirelson's bound~\cite{tsirel:original} for the Bell inequality~\cite{bell} known as the CHSH inequality~\cite{chsh}. 

\subsubsection{CHSH as a game}\label{sec:game}

We briefly describe the CHSH inequality in its more modern form as a game involving two distant players, Alice and Bob.
A detailed account of this formulation and how it allows us to recover the original form of the CHSH inequality can for example be found in~\cite{steph:diss}.
In the CHSH game, we imagine that we pose a question $y \in \01$ to Alice and a question $z \in \01$ to Bob, chosen uniformly at random, i.e., $p(y) = p(z) = 1/2$.
These questions can be identified with the choice of measurement setting in the usual formulation. Alice and Bob now return answers $a \in \01$ and $b \in \01$ respectively, where we say that Alice and Bob \emph{win} the game if and only if 
\begin{align}\label{eq:chshCondition}
	y \cdot z = a + b \mod 2\ .
\end{align}
Alice and Bob may thereby agree on any strategy beforehand, but they can no longer communicate once the game starts. In the quantum setting, this 
strategy corresponds to a choice of shared state and measurements, and in an experiment the no-signaling assumption is employed to enforce their
inability to communicate. Clearly, one may write the probability that Alice and Bob win for a particular strategy as
\begin{align}
	p_{\rm win} = \frac{1}{4} \sum_{y,z \in \01} \sum_{\substack{a,b\\a + b = y \cdot z}} \Pr[a,b|y,z]\ ,
\end{align}
where $\Pr[a,b|y,z]$ denotes the probability that Alice and Bob return answers $a$ and $b$ given questions $y$ and $z$.
For any classical strategy, $p_{\rm win} \leq 3/4$ but quantumly there exist a strategy that achieves $p_{\rm win} = 1/2 + 1/(2\sqrt{2}) \approx 0.853$.
This is in fact optimal, since Tsirelson has shown~\cite{tsirel:original,tsirel:separated}
that for any quantum strategy
\begin{align}\label{eq:tsirelBound}
	p_{\rm win} \leq \frac{1}{2} + \frac{1}{2\sqrt{2}}\ .
\end{align}

\subsubsection{Strategies and state discrimination}\label{sec:tsirelNotation}

For our purposes, it will be convenient to employ a simple observation about what Bob has to do in order to produce the right answer in the game, 
which was described in more detail in~\cite{steph:diss}.
Let $\rho_{y,a}$ denote the state of Bob's system conditioned on the fact that Alice received question $y$ and has given answer $a$.
Note that Bob's system will be placed in this state with probability $p(y,a) = p(a|y)/2$. 
For $z=0$, ~\eqref{eq:chshCondition} the rules of the game
state that Alice and Bob win if and only if Bob returns the same answer as Alice, that is, $b = a$. In other words, Bob would like to determine, which
of the following two states he is given
\begin{align}
	\sigma_0^{z=0} &:= \left(q^{z=0,0}_0\rho_{0,0} + q^{z=0,0}_1 \rho_{1,0}\right)\ ,\\
	\sigma_1^{z=0} &:= \left(q^{z=0,1}_0\rho_{0,1} + q^{z=0,1}_1\rho_{1,1}\right)\ ,
\end{align}
where 
\begin{align}
	q^{z=0,0}_y &= p(0|y)/(p(0|0) + p(0|1))\ ,\\
	q^{z=0,1}_y &= p(1|y)/(p(1|0) + p(1|1))\ ,
\end{align}
and the probability of $\sigma_x^{z=0}$ is given by $p_x^{z=0} = (p(x|0) + p(x|1))/2$.
That is, Bob would simply try to extract classical information stored
in quantum states, which is exactly the setting that our bound applies to. Producing a classical outcome on the ancilla system is very natural in this setting
as we can imagine that when giving his answer Bob simply returns his ancilla to a referee who decides whether Alice and Bob 
win~\footnote{Note that our assumption that the referee is unrestricted contrasts with the view of computational complexity in which such
games play a role in interactive proof systems. There, Alice and Bob are all-powerful, but the referee has limited time at his disposal to decide
the outcome of the game.
We would like to emphasize that our aim here is entirely different since we are merely interested in the strength of correlations between Alice and
Bob that can be obtained within a certain time frame.}.
Similarly, if $z=1$ Bob would like to determine which of the following
two states he is given
\begin{align}
	\sigma_0^{z=1} &:= \left(q^{z=1,0}_0\rho_{0,0} + q^{z=1,0}_1 \rho_{1,1}\right)\ ,\\
	\sigma_1^{z=1} &:= \left(q^{z=1,1}_0\rho_{0,1} + q^{z=1,1}_1\rho_{1,0}\right)\ ,
\end{align}
where 
\begin{align}
	q^{z=1,0}_{y=0} &= p(0|0)/(p(0|0) + p(1|1))\ ,\\
	q^{z=1,0}_{y=1} &= p(1|1)/(p(0|0) + p(1|1))\ ,\\
	q^{z=1,1}_{y=0} &= p(1|0)/(p(1|0) + p(0|1))\ ,\\
	q^{z=1,1}_{y=1} &= p(0|1)/(p(1|0) + p(0|1))\ ,
\end{align}
the probability of $\sigma_0^{z=1}$ is $p_0^{z=1} = (p(0|0) + p(1|1))/2$, and the probability
of $\sigma_1^{z=1}$ is $p_1^{z=1} = (p(1|0) + p(0|1))/2$. 
The probability that Alice and Bob win the game for a particular strategy can now be expressed as
\begin{align}
	p_{\rm win} = \frac{1}{2}\sum_{z \in \01} P_{\rm guess}(X^z|E^z)\ ,
\end{align}
where we write $P_{\rm guess}(X^z|E^z)$ for Bob's success probability in solving the state discrimination problems described above for $z \in \01$. 
From this perspective, Tsirelson's bound provides us with an upper bound on how well we can solve these two problems on average.

\subsubsection{A time limited game}
In the usual setting of this game, Alice and Bob are essentially given an unlimited amount of time and energy to produce their answers. But how well can they do given only a limited amount of energy and time? Here, we consider a time-limited version of the CHSH game, in which Alice and Bob are given a fixed time $t$ to produce 
their answers. If no answers are given at time $t$, we automatically rule that Alice and Bob loose. 
Our goal will be to derive a time-dependent version of~\eqref{eq:tsirelBound}. For simplicity, we will thereby assume that Alice has an essentially unlimited amount of energy at her disposal and only Bob will be restricted in some fashion. 
Given the perspective that Bob has to solve a state discrimination problem to produce the right answer as explained above, it is clear that we can use our
general bound to address this setting. The use of an ancilla register is very natural, 
as we can view it as the message system holding  Bob's answer that is returned to the referee.

In the usual scenario, Alice and Bob can choose which state to share at the start of the game as part of their strategy. 
Note, however, that we cannot allow arbitrary starting states to begin with, as we want to put a limit on the energy that Bob has at his disposal. 
For simplicity, however, we will make the sole assumption that Bob's Hamiltonian is bounded as $\|H\|_\infty$. In the appendix, we will derive a general
time dependent Tsirelson bound from this assumption where we will need our generalization of the time bound for two input states to the case of non-uniform
input distributions.

Here, we will focus on the essential idea that underlies this bound which already becomes apparent if we consider a slightly simpler scenario
in which Alice's marginal distributions are uniform ($p(a|y) = 1/2$ for all $y$). This scenario is well motivated if we imagine that
there is a 
source supplying Alice and Bob with the maximally entangled state which lies outside of their control, and 
their strategy is restricted to their choice of two-outcome observables.
In this case, Alice's outcome distribution will either be deterministic or uniform.
In the deterministic case,
Alice essentially plays a classical strategy. To obtain a quantum advantage in the case
of unlimited time, Alice's outcome distributions will be uniform, and we will hence focus on this case.

To obtain a time-dependent Tsirelson bound, we now employ our simple bound involving the original trace distance of the two states that we wish to discriminate~\eqref{eq:simpleTrace}. 
We have by Tsirelson's bound that
\begin{align}
	\frac{1}{2}\sum_{z \in \01} P_{\rm guess}(X^z|E^z) \leq \frac{1}{2} + \frac{1}{2\sqrt{2}}\ ,
\end{align}
and hence by~\eqref{eq:distStandard} and the fact that $p(a|y) = 1/2$
\begin{align}
	\frac{1}{2}\sum_{z \in \01} D(\sigma_0^z,\sigma_1^z) \leq \frac{1}{\sqrt{2}}\ ,
\end{align}
otherwise there would exist a better strategy for Alice and Bob at long times. So we have from~\eqref{eq:simpleTrace} that
\begin{align}
	p_{\rm win} \leq \frac{1}{2} + \frac{\gamma t \|H\|_\infty}{2\sqrt{2}\hbar}\ .
\end{align}
In particular, this means that if we allow only a limited amount of energy by Bob (e.g., by demanding that $\|H\|_\infty = 1$), then Bob needs time at least
\begin{align}\label{eq:mintime}
	t \geq \frac{\hbar}{\gamma \|H\|_\infty}\ 
\end{align}
to achieve the optimum quantum violation of CHSH.
Note that to achieve the optimum quantum violation, Alice's marginals will in fact be uniform, and hence this 
is indeed the minimum time required.

Clearly, for small time frames, it would be better for Alice and Bob to play a classical strategy in which Bob can just
return the ancilla $\ket{0}$ ''as is'' to the referee. The tradeoff betweeen the classical and quantum strategies in our setting
can be captured when considering arbitrary distributions, which we will address in the appendix. In particular, we will show that
\begin{corollary}
	Let Bob's Hamiltonian be scaled such that $H \geq 0$.
	Then the maximum success probability of winning the CHSH game for Alice
	and Bob in time $t$ obeys
	\begin{align}\label{eq:timeTsirel}
		p_{\rm win}^t  \leq \frac{3}{4} + \frac{\gamma t \|H\|_\infty}{\sqrt{2} \hbar}\ ,
	\end{align}
where
	\begin{align}
		\gamma := \left\{ \begin{array}{ll}5/\pi & \mbox{if } 1 < t E_n/\hbar < 4 \ ,\\
			3/\pi & \mbox{otherwise} \ .
		\end{array}\right.
	\end{align}
\end{corollary}
%The factor $3/4$ thereby stems from the fact that we now allow arbitrary outcome distributions for Alice,
%and they can achieve the classical bound of $3/4$ by Alice outputting '0', and Bob returning the ancilla
%of '0' to the referee.
We could also derive a more general bound in terms of Bob's average energy using~\eqref{eq:avg2States}.
However, such a bound does not compare easily to the original Tsirelon's bound.

Of course, the minimum time~\eqref{eq:mintime} is extremely small, and irrelevant for any practical tests of CHSH. Indeed, it is not our intention to question the validity of present CHSH experiments or suggest any loopholes caused by an insufficient distance for Alice and Bob compared to the time it takes them to achieve Tsirelson's bound.
Instead, we provided the present analysis as an illustrative example of how our bound applies. 

We would like to point out that~\eqref{eq:timeTsirel}
tells us that the strength of non-local correlation is indeed a function of time.
Furthermore,~\eqref{eq:mintime} tells us that there exists a fundamental time required to establish maximally strong quantum correlations. 
Finally, we note that one can also interpret~\eqref{eq:timeTsirel} in another way: Let's suppose that we were to fix a time $t$ and observe that Alice and Bob tend to win the game with probability at least $q$. We can now rewrite~\eqref{eq:timeTsirel} to obtain a lower bound on $\|H\|_\infty$. That is we can conclude that Bob had a certain energy at his disposal, and the strength of non-local correlations in this setting provides us with a form of ''energy witness'' for Bob. This also holds for the most general case discussed in the appendix.

\section{Discussion}

\subsection{Choice of initial state}
We obtained a series of simple bounds on how well we can recover classical information stored in a quantum system within a certain timeframe.
Let us now first consider what role the choice of initial state of the 
ancilla played in our bounds. During our discussions we assumed that the ancilla started out in the classical state corresponding to the most likely symbol $\xmax$. This reflects the fact that the distinguisher \emph{does} have full knowledge not only about the states $\rho_x$ themselves, but also about the distribution $p_x$. In particular, this means that without touching the quantum register, he can always achieve a success probability of $\pmax$ by outputting $\xmax$. 
Clearly, we could have chosen any other classical symbol as our starting point, and our bounds can easily be adapted accordingly. 
This holds even for an arbitrary pure state of the ancilla.
Yet, such a choice does not reflect the distinguisher's apriori knowledge.

Another option would be to let the ancilla start out in a special blank state, which intuitively corresponds to an outcome of ``don't know''.
and is orthogonal to any other outputs. It is straightforward to apply our methods to obtain a similar bound for this case.
Yet, note that using a blank ancilla state is conceptually rather different since it means that we essentially neglect the apriori
knowledge that a distinguisher has available. 

\subsection{Input size}

Our bound is especially useful, if we are merely concerned with the probability of success that can be achieved withing a certain time $t$ \emph{in principle}, 
using any physically allowed operation $H$. This is 
indeed interesting when we consider the problem of Bell inequalities where we wanted to obtain a bound on how well Alice and Bob can violated CHSH
within a given time frame, when they can \emph{choose} any Hamiltonian they like subject to energy constraints alone. 
In particular, we would like to emphasize that the time required to acquire classical information in our setting is not limited by the size of the alphabet $\setX$,
but merely by the choice of encodings. In practise, however, there are much more stringent constraints
on how quickly information can be transferred that depend on the geometry of the ancilla, leading to additional constraints on the interaction Hamiltonian $H$.
For example, it could be that $H$ can consist only of two qubit interactions, and interactions between the encoding system and the ancilla are limited to their boundary.
In this case, the size of the alphabet $\setX$ clearly \emph{does} matter, and stronger bounds therefore should depend strongly on the exact form of $H$.
We note that some bounds on time scales for particular Hamiltonians $H$ do follow from the decoherence and thermodynamics literature~\cite{gogolin:mthesis,lazyStates} for pure state 
encodings, yet since such bounds
typically involve $\|H_{\rm int}\|_\infty$, where $H_{\rm int}$ is the interacting part of $H$ they offer little advantage in our setting.
To see how such bounds are related to ours is most easily seen when considering the conditional von Neumann entropy $\HS(X|E)$. Note that if all $\rho_x$
are pure the overall cqq-state $\rho_{XEA}$~\footnote{A cqq-state is a classical-quantum-classical state, here classical on the source registers $X$, quantum on the encoding system $E$
and, before the referee's measurement, quantum on the ancilla $A$.} 
is pure as well. Hence, $\HS(X|E) = \HS(XE) - \HS(E) = \HS(A) - \HS(E)$. To determine how $\HS(X|E)$ can change with time we would thus like to determine
how the entropy of the reduced systems $A$ and $E$ evolves with time which has been studied for the von Neumann entropy in the decoherence literature where an upper bound for the rate of change in entropy was obtained in terms of $\|H_{\rm int}\|_\infty$~\cite{lazyStates}. Similar considerations can be made for other entropies~\cite{adrian:mthesis}.
It is an interesting open question to obtain good bounds on such quantities for arbitrary $H$ that take more of their structure into account. 

\subsection{Open questions}

Clearly, this is not the only interesting open question. Closely related is the question of how much time is required to demonstrate non-local correlations if Alice and Bob are yet more restricted. Again, this could take the form of physical constraints on the ancilla, or be considered in the framework of circuit complexity where one cares about the number of two qubit interactions, i.e., gates, that they have to apply. The example of CHSH is too small for such constraints to make a difference, but do play an important role when considering more complicated inequalities. 

Furthermore, it would be nice to see if the slight gap between our bound and the Margolus-Levitin theorem can be closed completely using a more stringent analysis for the case of orthogonal encodings $\rho_x$. In particular, this means that one would rederive the exact form of the Margolus-Levitin theorem from the rate of information transfer alone.

\acknowledgments

SW thanks Oscar Dahlsten, Artur Ekert, Christian Gogolin, Peter Janotta, Jonathan Oppenheim, Renato Renner, Thomas Vidick 
and CQT's ''non-local club'' for interesting comments and discussions. SW would also like to thank Artur Ekert and Jonathan Oppenheim for
useful pointers to existing literature~\cite{mlTheorem, bekenstein:survey}.
SW was supported by the National Research Foundation (Singapore), and the Ministry of Education (Singapore). ACD acknowledges support through the ARC Centre of Excellence in Engineered Quantum Systems (EQuS), project number CE110001013.

\appendix

\section{Basic observations}

In this appendix, we provide the technical details of our claims. 
To this end, we first establish two simple lemmas from which we later derive all our results. Since we will use these in everything that follows we consider the generalized problem
where we wish to distinguish $N$ states $\rho_0,\ldots,\rho_{N-1}$.
The first lemma will be used to bound the success probabilies
using measurement operators $M_x = \id \otimes P_x + W_x$ where the label $x \in \{0,\ldots,N-1\}$ corresponds to one of the $N$ states we wish to idenitfy. 

\begin{lemma}\label{lem:basicBound}
	For any Hermitian operator $A \in \bop(\hin)$ with diagonalization $A = \sum_j \lambda_j \proj{u_j}$, and any $x\in\{0,\ldots,N-1\}$ 
	the operator $\tilde{A} := A \otimes \pinit$ satisfies
	\begin{align}
		\tr\left(W_x\tilde{A}\right) \leq
			2 \sum_n (1 - \cos(t E_n/\hbar)) \bra{E_n}\tilde{A}_x\ket{E_n}\ ,
	\end{align}
		where
	\begin{align}
		\tilde{A}_x &:= \left\{ 
		\begin{array}{ll}
			|A| \otimes \pinit & {\rm for\ } x = \xmax\ ,\\
			A^+ \otimes \pinit
			& {\rm otherwise}\ ,
		\end{array}\right. \\
	A^+ &:= \sum_{j,\lambda_j \geq 0} \lambda_j \proj{u_j}\ .
\end{align}
\end{lemma}
\begin{proof}
	Using the definition of $W_x$ from~\eqref{eq:Wdef} we evaluate the terms involving $W_x^1$ and $W_x^2$ separately. 
	Let us now first bound the term involving $W_x^1$.
	For $x\neq \xmax$ we have that
	\begin{align}
		&\tr\left(W_x^1 \tilde{A}\right)\nonumber\\
		&\qquad=\tr\left((\id \otimes P_x)R \tilde{A}\right) + \tr\left(R^\dagger (\id \otimes P_x)\tilde{A}\right)\\
		&\qquad=0\ ,
	\end{align}
	where we used the linearity and cyclicity of the trace, as well as the fact that $P_x \pinit = \pinit P_x = 0$ for all $x \neq \xmax$.
	Let $A^- := \sum_{j,\lambda_j < 0} \proj{u_j}$, and define
	$\tilde{A}^+ := A^+ \otimes \pinit$ and $\tilde{A}^- := A^- \otimes \pinit$. Note that $\tilde{A} = \tilde{A}^+ - \tilde{A}^-$.
	For $x=\xmax$ we can now use the fact that
	\begin{align}\label{eq:RDef}
		R &= U - I = \sum_n (\exp(-itE_n/\hbar) - 1) \proj{E_n}\ ,
	\end{align}
	to write
	\begin{align}
		&\tr\left(W_\xmax^1 \tilde{A}\right)\\
		&= \tr\left((\id \otimes P_\xmax)R \tilde{A}\right) + \tr\left(R^\dagger (\id \otimes P_\xmax)\tilde{A}\right)\nonumber\\
		&=\tr\left( (R + R^\dagger)\tilde{A}\right)\\
		&= \sum_{n=0}^\infty (\exp(-itE_n/\hbar) +\exp(itE_n/\hbar) - 2)\nonumber\\
		&\qquad\qquad\qquad
\bra{E_n}\tilde{A}\ket{E_n}\\
&=2 \sum_{n=0}^{\infty} (\cos(t E_n/\hbar) - 1)  \bra{E_n}\tilde{A}\ket{E_n}\ .\\
&=2 \sum_{n=0}^{\infty} (\cos(t E_n/\hbar) - 1)  \bra{E_n}(\tilde{A}^+-\tilde{A}^-)\ket{E_n}\\
&\leq 2 \sum_{n=0}^{\infty} (1 - \cos(t E_n/\hbar))  \bra{E_n}\tilde{A}^-\ket{E_n}\ ,
	\end{align}
	where the fourth equality follows from Euler's formula, 
	and the first inequality from the fact that $\cos(t E_n/\hbar) - 1 \leq 0$ and $\tilde{A}^+,\tilde{A}^- \geq 0$.

	It remains to bound the term involving $W_x^2$.
	First of all, since $\Lambda(X) = R X R^\dagger$ is a positive map~\cite{bathia:posBook}, and $\tilde{A}^+,\tilde{A}^- \geq 0$, we have
	that 
	\begin{align}\label{eq:posTerms}
		R\tilde{A}^+R^\dagger &\geq 0\ ,\\
		R\tilde{A}^-R^\dagger &\geq 0\ .
	\end{align}
	Note that for any $X,Z \geq 0$, we have $\tr(XZ) \geq 0$, and hence
	$\tr\left( (\id \oplus P_x) R\tilde{A}^-R^\dagger\right) \geq 0$.
	Second, note that $RR^\dagger = R^\dagger R$ and we have
	\begin{align}\label{eq:RR}
		&RR^\dagger = \\
		&= \sum_{n=0}^\infty \left(2 - \exp(i t E_n/\hbar) - \exp(-i t E_n/\hbar)\right)\nonumber\\
		&\qquad\qquad\qquad\proj{E_n} \\
		&= 2 \sum_n (1 - \cos(t E_n/\hbar)) \proj{E_n}\ ,
	\end{align}
	where the second equality follows by applying Euler's formula.
	We thus have
	\begin{align}
		&\tr\left(W_x^2 \tilde{A}\right)\\
		&=
		\tr\left((\id \otimes P_x)R \tilde{A}^+ R^\dagger\right) - 
		\tr\left((\id \otimes P_x)R \tilde{A}^- R^\dagger\right)\nonumber\\
		&\leq \tr\left(R^\dagger R \tilde{A}^+\right)\\
		&= 2 \sum_n (1 - \cos(t E_n/\hbar)) \bra{E_n}\tilde{A}^+\ket{E_n}
	\end{align}
	where the first inequality follows from~\eqref{eq:posTerms}, the fact that
	$0 \leq \id \otimes P_x \leq \id$ and the cyclicity of the trace, and the last equality from~\eqref{eq:RR}.
	Putting everything together, $\tr(W_x\tilde{A}) = \tr(W_x^1\tilde{A}) + \tr(W_x^2\tilde{A})$, we obtain the claimed result.
\end{proof}

We will also make repeated use of the following bound. Note that whereas the bound applies to a very large range of values $E_n \geq 0$, we will later be particularly
interested in the case of $t E_n/\hbar < 1$. Indeed the bound below is a great overestimate
if $t E_n/\hbar > 2\pi$, as $2(1-\cos(k)) \leq \gamma (k - 2 \pi \floor{k/(2\pi)})$.

\begin{lemma}\label{lem:gammaBound}
	Let $E_n \geq 0$. Then $2(1-\cos(t E_n/\hbar)) \leq \gamma t E_n/\hbar$ where
	\begin{align}
		\gamma := \left\{ \begin{array}{ll}5/\pi & \mbox{if } 1 < t E_n/\hbar < 4\ ,\\
			3/\pi & \mbox{otherwise}\ .
		\end{array}\right. 
	\end{align}
\end{lemma}

\section{A bound for two states}

We now first consider the case were we are given just two states, $\rho_0$ and $\rho_1$. Here, we will consider the most general problem
where $p_0$ and $p_1$ can be arbitrary.

\subsection{A bound in terms of the trace distance}

First of all, note that even for a general distribution $\{p_x\}_x$ the problem of distinguishing two states is easy to analyze~\cite{helstrom:detection}.
In particular, we have that in the time-unlimited case for measurement operators acting directly on the encoding space
\begin{align}\label{eq:unequalProb}
	P_{\rm guess}(X|E) &= \max_{M_0,M_1} p_0 \tr(M_0 \rho_0) + p_1 \tr(M_1 \rho_1)\\
	&= p_0 + \max_{M_1} \tr\left(M_1(p_1 \rho_1 - p_0 \rho_0)\right)\\
	&= p_0 + \Delta(p_1 \rho_1, p_0 \rho_0)\ ,
\end{align}
where $\Delta(p_1 \rho_1, p_0 \rho_0)$
is given by 
\begin{align}
	\Delta(p_1 \rho_1,p_0 \rho_0) &= \max_{0 \leq P \leq \id} \tr\left(PA\right)\\
	&=\tr(A^+)\ , \label{eq:DeltaTrace}
\end{align}
where $A := p_1 \rho_1 - p_0 \rho_0$ with diagonalization $A = \sum_j \lambda_j \proj{u_j}$ and
$A^+ = \sum_{j,\lambda_j \geq 0}\proj{u_j}$ 
(Note that $\Delta$ is not symmetric here and hence formally does not form
a distance measure.) Similarly, we have
\begin{align}
	\Delta(p_0\rho_0,p_1\rho_1) = \tr(A-)\ .
\end{align}
Note that for the time unlimited case, we could have equivalently expressed
the success probability
as
\begin{align}
	P_{\rm guess}(X|E) = p_1 + \Delta(p_0\rho_0,p_1\rho_1)\ .
\end{align}
It will also be useful to note that for $\tilde{\rho}_x = \rho_x \otimes \pinit$,
\begin{align}\label{eq:ancillaDoesntChangeDelta}
\Delta(p_1 \rho_1,p_0 \rho_0) 
= \Delta(p_1 \tilde{\rho}_1,p_0 \tilde{\rho}_0)\ .
\end{align}

Before stating our bound, let us introduce some additional notation. For two states, define
\begin{align}
	\xmin := 1 - \xmax\ .
\end{align}
We now first relate the problem of discriminating the two states in time $t$
to the original success probability.

\begin{lemma}\label{lem:proto2}
	The probability of distinguishing $\rho_0$ and $\rho_1$ given with probabilities $p_0$ and $p_1$ 
	using the Hamiltonian $H = \sum_n E_n \proj{E_n} \geq 0$ is bounded by
	\begin{align}
		&P_{\rm succ}(X|E)_{H,t}
		\leq \pmax + \\
		&\qquad 2(1 - \cos(t C_{\rm max}/\hbar)) \Delta(\pmin \rhomin,\pmax \rhomax)\ ,\nonumber
	\end{align}
	where $C_{\rm max} = \argmax_{E_n} (1 - \cos(t E_n/\hbar))$.
\end{lemma}

\begin{proof}
	Using~\eqref{eq:unequalProb},~\eqref{eq:measurementOperators} and the fact that $\Pmin\init = 0$
	we may bound the success probability as 
	\begin{align}\label{eq:rewriteProb}
		&P_{\rm succ}(X|E)_{H,t} \leq \pmax +\\
		&\qquad \tr\left(\Wmin(\pmin \tilde{\rho}_{\xmin} - \pmax \tilde{\rho}_{\xmax})\right)\ .\nonumber
	\end{align}
	Applying Lemma~\ref{lem:basicBound} for $A = \pmin \rhomin - \pmax \rhomax$ we have that
	\begin{align}
		&\tr\left(\Wmin(A\otimes \pinit)\right)\\
		&\leq 2 \sum_n (1 - \cos(t E_n/\hbar)) \bra{E_n}A^+ \otimes \pinit\ket{E_n}\ .\nonumber
	\end{align}
	Hence from~\eqref{eq:DeltaTrace} and~\eqref{eq:ancillaDoesntChangeDelta} we have
	\begin{align}
		&\tr\left(\Wmin(\pmin \tilde{\rho}_{\xmin} - \pmax \tilde{\rho}_{\xmax})\right)\\
		&\leq 2(1-\cos(t C_{\rm max}/\hbar)) \Delta(\pmin \rhomin,\pmax \rhomax)\ . \nonumber
	\end{align}
	Our claim now follows by plugging this bound into~\eqref{eq:rewriteProb}.
\end{proof}

With the help of Lemma~\ref{lem:gammaBound} one may now also use the fact that $\forall E_n, E_n \leq \|H\|_\infty$ to 
obtain a very simple bound in terms of the spectrum of the Hamiltonian.
\begin{corollary}\label{cor:traceDist}
	The probability of distinguishing $\rho_0$ and $\rho_1$ given with probabilities $p_0$ and $p_1$ 
	using the Hamiltonian $H = \sum_n E_n \proj{E_n} \geq 0$ is bounded by
	\begin{align}
		&P_{\rm succ}(X|E)_{H,t} \leq \pmax +\\
		&\qquad \frac{\gamma t \|H\|_\infty \Delta(\pmin \rhomin, \pmax \rhomax)}{\hbar}\ .\nonumber
	\end{align}
\end{corollary}

\subsection{A bound in terms of the average energy}

Inspecting the proof above with Lemma~\ref{lem:gammaBound} in mind, it is indeed easy to see that we can also obtain a bound in terms
of average energies. We first derive a somewhat stronger bound for two equiprobable states that actually depends on the ``average energy'' of a function of both states.

\begin{theorem}\label{thm:avgEnergyTwoStates}
	The probability of distinguishing $\rho_0$ and $\rho_1$ given with probabilities $p_0$ and $p_1$ using the Hamiltonian $H = \sum_n E_n \proj{E_n} \geq 0$ in time $t$ is bounded by
	\begin{align}
		&P_{\rm succ}(X|E)_{H,t}
		\leq \pmax +\\
		&\qquad\frac{\gamma t}{2\hbar} \left[\tr(H|\pmin \tilde{\rho}_{\xmin}- \pmax\tilde{\rho}_{\xmax}|)\right. + \nonumber\\
		&\qquad \left. \pmin \tr(H\tilde{\rho}_{\xmin}) - \pmax \tr(H\tilde{\rho}_{\xmax})\right]\ ,\nonumber
	\end{align}
	where
	\begin{align}
		\gamma := \left\{ \begin{array}{ll}5/\pi & \mbox{if } 1 < t E_n/\hbar < 4 \ ,\\
			3/\pi & \mbox{otherwise} \ .
		\end{array}\right.
	\end{align}
\end{theorem}
\begin{proof}
	Recall that by applying Lemma~\ref{lem:basicBound} with $A = \pmin \rhomin - \pmax \rhomax$ we have that
	\begin{align}
		&\tr\left(\Wmin(\pmin \tilde{\rho}_{\xmin} - \pmax \tilde{\rho}_{\xmax})\right) \\
		&\leq 2 \sum_n (1 - \cos(t E_n/\hbar)) \bra{E_n}A^+ \otimes \pinit\ket{E_n}\ .\nonumber
	\end{align}
	We may now use Lemma~\ref{lem:gammaBound} to obtain
	\begin{align}
		&\tr\left(\Wmin(\pmin \tilde{\rho}_{\xmin} - \pmax \tilde{\rho}_{\xmax})\right) \\
		&\leq \frac{\gamma t}{\hbar} \sum_n E_n \bra{E_n}A^+ \otimes \pinit\ket{E_n}\nonumber\\
		&= \frac{\gamma t \tr(H\tilde{A}^+)}{\hbar}\ .
	\end{align}
	Our claim now follows by noting that 
	\begin{align}
		\tilde{A}^+ &= \frac{1}{2}\left(\tilde{A}^+ + \tilde{A}^{-}\right) + \frac{1}{2}\left(\tilde{A}^+ - \tilde{A}^-\right)\\
		&=\frac{1}{2}\left|\pmin \tilde{\rho}_{\xmin} - \pmax \tilde{\rho}_{\xmax}\right| +\\
		&\qquad\frac{1}{2}\left(\pmin \tilde{\rho}_{\xmin} - \pmax \tilde{\rho}_{\xmax}\right)\ .\nonumber
	\end{align}
\end{proof}

\section{A bound for many input states}

Finally, we derive a bound for the most general case of distinguishing states $\rho_0,\ldots,\rho_{N-1}$ where we are given $\rho_x$ with 
probability $p_x$.

\begin{theorem}\label{thm:manyStates}
	Suppose $H \geq 0$. Then the probability of distinguishing $\rho_0,\ldots,\rho_{N-1}$ given with probabilities
	$p_0,\ldots,p_{N-1}$ obeys
	\begin{align}
		P_{\rm succ}(X|E)_{H,t}
		&\leq \pmax +\frac{\hat{\gamma}t}{\hbar} \sum_{x=0}^{N-1} p_x \tr\left(H\tilde{\rho}_x\right)\ ,
	\end{align}
	where
	\begin{align}
		\hat{\gamma} := \left\{ \begin{array}{ll}5/\pi & \mbox{if } \forall E_n, 1 < t E_n/\hbar < 4\\
			3/\pi & \mbox{otherwise}
		\end{array}\right. \ .
	\end{align}
\end{theorem}
\begin{proof}
Note that the success probability for a particular interaction $H$ is now given by
\begin{align}\label{eq:pSUCCGeneral}
	&P_{\rm succ}(X|E)_{H,t} = \sum_{x=0}^{N-1} p_x \tr\left(M_x \tilde{\rho}_x\right)\ ,
\end{align}
where
\begin{align}
\tr\left(M_x \tilde{\rho}_x\right)=
\tr\left( ( \id \otimes P_x) \tilde{\rho}_x\right) + \tr\left(W_x \tilde{\rho}_x\right)\ .
\end{align}
Let us now first consider the case of $x = \xmax$. 
We have that
\begin{align}
\tr\left((\id \otimes P_x)\tilde{\rho}_x\right) = \tr(\rho_x) = 1\ .
\end{align}
Using Lemma~\ref{lem:basicBound} with $A = \rho_x$ we hence have that
\begin{align}
&\tr\left(M_x \tilde{\rho}_x\right) \leq\\
&\qquad 1 + 2  \sum_n (1-\cos(t E_n/\hbar)) \bra{E_n}\tilde{\rho}_x\ket{E_n}\nonumber\ .
\end{align}
We now turn to the case of $x \neq \xmax$.
Since $P_x\pinit = \pinit P_x = 0$ and $\tilde{\rho}_x = \rho_x \otimes \pinit$ we have
$(\id \otimes P_x)\tilde{\rho}_x = 0$ for all $x \neq \xmax$.
	Again by applying~\ref{lem:basicBound} with $A = \rho_x$ we obtain
	from $\rho_x \geq 0$ that
	\begin{align}
		\tr\left(W_x \tilde{\rho}_x\right) \leq 2 \sum_n (1-\cos(t E_n/\hbar)) \bra{E_n}\tilde{\rho}_x\ket{E_n}\ .
	\end{align}
	Our claim now follows by using Lemma~\ref{lem:gammaBound} to obtain $2(1 - \cos(t E_n/\hbar)) \leq \gamma t E_n/\hbar$
	for $E_n \geq 0$.
\end{proof}

%Let us now consider the case where the referee initializes the ancilla to a random state $\proj{x}$ with probability $p_x$. Note that this corresponds to a situation where the distribution on the ancilla essentially encodes our ignorance prior to any measurement. 
%From the perspective of the distinguisher the ancilla is mixed, but from the perspective of the referee who has knowledge of the ancilla it is indeed pure. 
%Using the same techniques as above, where we follow the proof of Lemma~\ref{lem:basicBound} with $x=j$ when the ancilla is in state $\proj{x}$ in place of $x=0$. We now have
%\begin{corollary}
%Suppose that the referee initializes the ancilla in state $\proj{x}$ with probability $p_x$.
%Furthermore, suppose $H \geq 0$. Then the probability of distinguishing $\rho_0,\ldots,\rho_{N-1}$ given with probabilities
	%$p_0,\ldots,p_{N-1}$ obeys
	%\begin{align}
		%&P_{\rm succ}(X|E)_{H,t} \\
		%&\leq 2^{-H_2(\{p_x\}_x)} + \frac{N-1}{N}\frac{\hat{\gamma}t}{\hbar} \sum_{x=0}^{N-1} p_x \tr\left(H\tilde{\rho}_x\right)\ ,\nonumber
	%\end{align}
	%where
	%\begin{align}
		%\hat{\gamma} := \left\{ \begin{array}{ll}5/\pi & \mbox{if } \forall E_n, 1 < t E_n/\hbar < 4\\
			%3/\pi & \mbox{otherwise}
		%\end{array}\right. \ ,
	%\end{align}
%and $H_2(\{p_x\}_x) = - \log \sum_x p_x^2$ is the collision entropy of the distribution $\{p_x\}_x$.
%\end{corollary}

\section{Attaining the bound}

We now exhibit a Hamiltonian that achieves our upper bound for the success probability of distinguishing two states given with apriori equal probability.

\begin{theorem}\label{thm:attaining}
	Suppose we are given $\rho_0$ and $\rho_1$ with apriori equal probability. Let $E_{\max} \geq 0$. 
	Then there exist a Hamiltonian $H$ with $\|H\|_\infty = E_{\max}$ that in time $t$ achieves success probability
	\begin{align}
		P_{\rm succ}(X|E)_{H,t} = \frac{1}{2} + \frac{1}{4}(1-\cos(t\|H\|_\infty /\hbar)) D(\rho_0,\rho_1)\ .
	\end{align}
	In particular, we can distinguish the two states perfectly in time $t = \hbar \pi/E_{\max}$.
\end{theorem}
\begin{proof}
	Let $A = \rho_1 - \rho_0$. We can diagonalize $A = \sum_j \lambda_j \proj{u_j}$, and define $A^+ := \sum_{j,\lambda_j \geq 0} \lambda_j \proj{u_j}$ 
	and $A^- := A - A^+$.
	Consider the operator
\begin{align}
	\hat{H} &:= \Pi_{A^-} \otimes \id 
	+\Pi_{A^+} \otimes (\outp{0}{1} + \outp{1}{0})\ ,
\end{align}
where $\Pi_{A^+}$ and $\Pi_{A^-}$ are projectors on the support of $A^+$ and $A^-$ respectively.
Clearly, $\hat{H}$ is Hermitian and unitary, and hence has eigenvalues $\pm 1$.
We now define the Hamiltonian $H$
\begin{align}
	H := E_{\max} (\hat{H} + \id)/2\ .
\end{align}
Since any term in the Hamiltonian proportional to the identity does not affect the dynamics, we may replace the time evolution operator $\exp(-itH/\hbar)$ with
\begin{align}
	U = \exp(-i t\|H\|_\infty\hat{H}/2\hbar) = \sum_{n=0}^{\infty}\frac{1}{n!} \left(\frac{-i t\|H\|_\infty\hat{H}}{2\hbar}\right)^n\ .
\end{align}
For our choice of $\hat{H}$ we have that 
\begin{align}
	(\hat{H})^n &=\left\{ \begin{array}{ll} \id &n \mbox{ even}\ ,\\
		\hat{H} &n \mbox{ odd}\ .\end{array}\right. 
\end{align}
and as a result the Taylor expansion for $U$ gives 
\begin{align}
	U&=(\cos(t\|H\|_\infty/2\hbar)\id - i \sin(t\|H\|_\infty/2\hbar)\hat{H})\ .
\end{align}
Using this unitary in our state discrimination problem, we obtain
\begin{align}
	&P_{\rm succ}(X|E)_{H,t}\\
	&= \frac{1}{2} + \frac{1}{2} \tr\left(U^\dagger (\id \otimes P_1) U(\tilde{\rho}_1 - \tilde{\rho}_0)\right)\nonumber\\
	&=\frac{1}{2} + \frac{\sin^2(t\|H\|_\infty/2\hbar)}{2} \tr\left(\hat{H}(\id\otimes P_1)\hat{H}(\tilde{\rho}_1 - \tilde{\rho}_0)\right)\ .
\end{align}
It remains to evaluate the last term.
First of all, note that
\begin{align}
	\hat{H}(\id \otimes P_1)\hat{H} &= (\Pi_{A^+} \otimes \proj{0} + \Pi_{A^-} \otimes \proj{1})\ .
\end{align}
Since $\tilde{\rho}_1 - \tilde{\rho}_0 = (A^+ - A^-) \otimes \proj{0}$ and $\Pi_{A^+}\Pi_{A^-} = \Pi_{A^-}\Pi_{A^+} = 0$, we thus have
\begin{align}
	\tr\left(\hat{H}(\id \otimes P_1)\hat{H}(\tilde{\rho}_1 - \tilde{\rho}_0)\right) &=\tr(A^+)\\
	&= D(\rho_0,\rho_1)\ .
\end{align}
The claim follows by an application of the double angle formula.
\end{proof}

\section{Constraining the eigenvalues of $H$}

For completeness, we now remind ourselves why in many settings it is not unreasonable to assume that $\|H\|_\infty$ is indeed bounded.
Note that when dealing with fixed input states $\rho_0,\ldots,\rho_{N-1}$, we can without loss of generality assume that the Hamiltonian
$H$ that leads to the optimal success probability possible within a certain time $t$ is limited to the energy eigenspace sufficient
to contain the support of the inputs states and the standard ancilla state. This holds even in an approximate sense.
To see this, consider how the state on the encoding register and the ancilla evolve in time 
\begin{align}
	\tilde{\rho}_0(t) := U(t) \tilde{\rho}_0 U(t)^\dagger\ ,\\
	\tilde{\rho}_1(t) := U(t) \tilde{\rho}_1 U(t)^\dagger\ .
\end{align}
Choose an error parameter $\eps$ and define $\eproj$ to be the lowest rank operator such that $[\eproj,H] = 0$ and 
\begin{align}\label{eq:epsClose}
	\frac{1}{2}\|\rho_0(0) - \eproj \rho_0(0) \eproj\|_1 &\leq \eps\ ,\\
	\frac{1}{2}\|\rho_1(0) - \eproj \rho_1(0) \eproj\|_1 &\leq \eps\ .
\end{align}
Since $[\eproj,H] = 0$ and the L1-norm is unitarily invariant, we can immediately conclude that~\eqref{eq:epsClose} still holds
when we replace $\rho_0(0)$ and $\rho_1(0)$ with any subsequent states $\rho_0(t)$ and $\rho_1(t)$.
However, this tells us that we can approximate $U$ with a unitary $\hat{U}$ as
\begin{align}
	\hat{H} &:= \eproj H \eproj\ ,\\
	\hat{U}(t) &:= \exp(- i \hat{H} t/\hbar)\ ,
\end{align}
without affecting any of the output states, except with a chosen error $\eps$. By the definition of the L1-norm 
we have for any Hermitian operator $A$ that
\begin{align}
	\frac{1}{2}\|A\|_1 = \sup_{-\id \leq P \leq \id}\tr(PA)\ ,
\end{align}
and hence~\eqref{eq:epsClose} implies that using the unitary $\hat{H}$ in place of any original $H$ leads to a change in success probability
in the state discrimination problem of at most $2 \eps$. 

\section{A general time-dependent Tsirelson's bound}

Let us now consider a more general version of our time-dependent Tsirelson's bound in which we drop the assumption
that the source emits a particular state, and that Alice makes a two-outcome projective measurement. The only assumption we will make
now is that Bob's Hamiltonian is bounded $\|H\|_\infty = E_{\max}$. 
%Note for such a Hamiltonian we may as well assume that $\|H\|_\infty = 1$
%as the addition of an identity term does not affect time scales.

For our proof, we will need the more general version of the two state discrimination problem in which the two states
are not necessarily given with equal probabilities (see Corollary~\ref{cor:traceDist}). Again, let us first briefly consider the
time unlimited case, where $M_0$ and $M_1$ are just measurements on a single system. Recall that we could express
the success probability of distinguishing $\rho_0$ and $\rho_1$ given with probabilities $p_0$ and $p_1$ respectively as
\begin{align}
	P_{\rm guess}(X|E) = p_0 + \Delta(p_1 \rho_1, p_0 \rho_0)\ .
\end{align}
At first glance, this expression appears a bit assymetric - after all, what should be so special about $p_0$? Note, however,
that by replacing $M_1 = \id - M_0$ in~\eqref{eq:unequalProb} we could also have expressed the success probability as
\begin{align}
	P_{\rm guess}(X|E)
	= p_1 + \Delta(p_0 \rho_0, p_1 \rho_1)\ . \nonumber
\end{align}
In particular, it will be convenient to note that we could have also written the success probability as the average of these two terms
\begin{align}
	P_{\rm guess}(X|E)
	= \frac{1}{2}\left(1 + \sum_{x\in\01}\Delta(p_{\bar{x}} \rho_{\bar{x}}, p_x \rho_x)\right)\ .
\end{align}
Let us now return to the time limited case, involving an interaction of the encoding and ancilla system, followed by a measurement
on the ancilla. Recall that we have from Corollary~\ref{cor:traceDist} that
\begin{align}\label{eq:timeLimitedSym}
	&P_{\rm guess}(X|E)_{H,t}\\
	&\leq p_{\xmax} + \left(\frac{t \gamma \|H\|_\infty}{\hbar}\right) \Delta(\pmin \rhomin,\pmax \rhomax)\ . \nonumber
\end{align}
Note that in the time limited case we cannot simply average -- the proof of Corollary~\ref{cor:traceDist} yields a different bound
had we placed $p_{\xmin}$ in front (a small calculation shows that it will again single out $\pmax$). 
We are now ready to show our general bound, where we will use the notation developed in Section~\ref{sec:tsirelNotation}. 
%Note that this time 
%we cannot make any assumptions about the distributions $p(a|s)$.

\begin{lemma}
	Let Bob's Hamiltonian be scaled such that $H \geq 0$.
	Then the maximum success probability of winning the CHSH game for Alice
	and Bob in time $t$ obeys
	\begin{align}
		p_{\rm win}^t  \leq \frac{1}{2}\left(\sum_{z\in\01} \pmax^z\right) + \frac{\gamma t \|H\|_\infty}{\sqrt{2} \hbar}\ ,
	\end{align}
where
	\begin{align}
		\gamma := \left\{ \begin{array}{ll}5/\pi & \mbox{if } 1 < t E_n/\hbar < 4 \ ,\\
			3/\pi & \mbox{otherwise} \ .
		\end{array}\right.
	\end{align}
\end{lemma}
\begin{proof}
We have from Tsirelson's bound~\cite{tsirel:original} that for any strategy of Alice and Bob no matter how much time or energy they may have available
\begin{align}
	p_{\rm win} &= \frac{1}{2}\sum_{z \in \01} P_{\rm guess}(X^z|E^z)\\
	&= \frac{1}{4}\sum_{z} \left(1 + \sum_{x \in \01} \Delta(p_{\bar{x}}^z \sigma_{\bar{x}}^z, p_x^z \sigma_x^z)\right)\\
	&\leq \frac{1}{2} + \frac{1}{2\sqrt{2}}\ .
\end{align}
Rearranging terms gives us
\begin{align}
	\frac{1}{4}\sum_{z,x\in \01} \Delta(p_{\bar{x}}^z \sigma_{\bar{x}}^z, p_x^z \sigma_x^z)
	&\leq \frac{1}{2} + \frac{1}{2\sqrt{2}} - \frac{1}{2}\\
	& = \frac{1}{2\sqrt{2}}\ .
\end{align}
Since $\Delta(\cdot,\cdot) \geq 0$, this means that
\begin{align}
	\frac{1}{2}\sum_{z} \Delta(\pmin^z \sigma_{\xmin}^z, \pmax^z \sigma_{\xmax}^z)
	\leq \frac{1}{\sqrt{2}}\ .
\end{align}
Our claim for the time limited case now follows by plugging this bound into~\eqref{eq:timeLimitedSym} 
\begin{align}
	p_{\rm win}^t &= \frac{1}{2}\sum_{z \in \01} P_{\rm guess}(X^z|E^z)_{H,t} \\
	&\leq \frac{1}{2}\left(\sum_{z \in \01} \pmax^z\right) +\\
	&
	\qquad\frac{C}{2} \sum_{z,x \in \01}
	\Delta(\pmin^z \sigma_{\xmin}^z, \pmax^z \sigma_{\xmax}^z)\ ,\nonumber\\
	&\leq \frac{1}{2}\left(\sum_{z \in \01} \pmax^z\right) + \frac{C}{\sqrt{2}}
\end{align}
where we have used the shorthand
\begin{align}
	C := \frac{t \gamma \|H\|_\infty}{\hbar}\ .
\end{align}
\end{proof}

At first glance, this bound may seem somewhat strange as it involves a potentially unknown term $\sum_x \pmax^z$.
Note, however, that this term is determined by the distributions over the states that Bob should distinguish. 
It is this distribution, that determines the classical bound for CHSH and hence
\begin{align}
	\frac{1}{2}\sum_z \pmax^z \leq \frac{3}{4}\ .
\end{align}
Note that since the ancilla is initialized to the most likely $x$ in each case, 
there exists a 
strategy for Alice and Bob
with which they can play optimally classically in no time at all. 
Yet, since there is generally an interplay between the choice of distributions and the state
Alice can create we derived the bound in its more general form above. 

\begin{corollary}\label{cor:tsirel}
	Let Bob's Hamiltonian be scaled such that $H \geq 0$.
	Then the maximum success probability of winning the CHSH game for Alice
	and Bob in time $t$ obeys
	\begin{align}
		p_{\rm win}^t  \leq \frac{3}{4} + \frac{\gamma t \|H\|_\infty}{\sqrt{2} \hbar}\ ,
	\end{align}
where
	\begin{align}
		\gamma := \left\{ \begin{array}{ll}5/\pi & \mbox{if } 1 < t E_n/\hbar < 4 \ ,\\
			3/\pi & \mbox{otherwise} \ .
		\end{array}\right.
	\end{align}
\end{corollary}

\end{document}